\documentclass[twocolumn]{article}

\usepackage{amsmath,amssymb,amsfonts,amsthm,amsbsy,mathtools}
\usepackage{mathrsfs}
\usepackage{tikz}
\usepackage{algorithm}
\usepackage{algorithmicx}
\usepackage{algpseudocode}
\usepackage{chngcntr}
\usepackage{titlesec}
\usepackage[hidelinks]{hyperref}

\newtheorem{theorem}{Theorem}
\numberwithin{theorem}{section}
\newtheorem{corollary}[theorem]{Corollary}
\newtheorem{proposition}[theorem]{Proposition}

\usepackage{standalone}
\standaloneconfig{mode=buildnew}

\date{}

\title{Change of basis for $\mathfrak{m}$-primary ideals in one and two variables}

\author{
Seung Gyu Hyun\thanks{University of Waterloo, Waterloo, ON, Canada (kevin.hyun@edu.uwaterloo.ca)}
\and
Stephen Melczer\thanks{University of Pennsylvania, 209 S. 33rd Street, Philadelphia, PA, USA (smelczer@sas.upenn.edu)}
\and
\'Eric Schost\thanks{University of Waterloo, Waterloo, ON, Canada (eschost@uwaterloo.ca)}
\and
Catherine St-Pierre\thanks{University of Waterloo, Waterloo, ON, Canada (catherine.st-pierre@uwaterloo.ca)}
}

\def\M{\mathsf{M}}
\def\A{\mathbb{A}}
\def\F{\mathbb{F}}
\def\Q{\mathbb{Q}}
\def\K{\mathbb{K}}
\def\L{\mathbb{L}}

\def\N{\mathbb{N}}
\def\frakm{\mathfrak{m }}
\def\frakn{\mathfrak{n }}

\newcommand{\factor}{0.45}

\newcommand{\muone}{2cm*\factor}
\newcommand{\mutwo}{1cm*\factor}
\newcommand{\muthree}{0cm*\factor}
\newcommand{\nuone}{0cm*\factor}
\newcommand{\nutwo}{1cm*\factor}
\newcommand{\nuthree}{2cm*\factor}
\newcommand{\rOffset}{6cm*\factor}

\newcommand{\xonestart}{4cm*\factor}
\newcommand{\xonemiddle}{2cm*\factor}
\newcommand{\xoneend}{0cm*\factor}
\newcommand{\xtwostart}{0cm*\factor}
\newcommand{\xtwomiddle}{1cm*\factor}
\newcommand{\xtwoend}{2cm*\factor}

\begin{document}
\maketitle

\begin{abstract}
Following recent work by van der Hoeven and Lecerf (ISSAC 2017), we
discuss the complexity of linear mappings, called \emph{untangling} and
\emph{tangling} by those authors, that arise in the context of computations
with univariate polynomials. We give a slightly faster tangling
algorithm and discuss new applications of these techniques. We 
show how to extend these ideas to bivariate settings, and use them to give bounds on the  arithmetic complexity of certain
algebras. 
\end{abstract}

\setcounter{secnumdepth}{4}

\titleformat{\paragraph}[runin]{\normalfont\normalsize\bfseries}{\theparagraph.}{1em}{}
\titlespacing*{\paragraph}{0em}{1ex}{1em}
\newcommand{\pref}[1]{{\bf\ref{#1}}}

\setcounter{secnumdepth}{4}
\renewcommand{\theparagraph}{%
    \ifnum\value{subsection}=0
      \thesection
    \else
      \thesubsection
    \fi
  .\arabic{paragraph}}

\alglanguage{pseudocode}
\renewcommand{\algorithmicrequire}{\textbf{Input:}}
\renewcommand{\algorithmicensure}{\textbf{Output:}}
\algnewcommand{\IIf}[1]{\State\algorithmicif\ #1\ \algorithmicthen}
\algnewcommand{\EndIIf}{\unskip\ \algorithmicend\ \algorithmicif}
\algnewcommand{\IfThenElse}[3]{\State \algorithmicif\ #1\ \algorithmicthen\ #2\ \algorithmicelse\ #3}


\section{Introduction}

In~\cite{HoLe17}, van der Hoeven and Lecerf gave
algorithms for ``modular composition'' modulo powers of polynomials:
that is, computing $F(G) \bmod T^\mu$, for polynomials $F,G,T$ over a
field $\F$ and positive integer $\mu$. As an intermediate result, they
discuss a linear operation and its inverse, which they
respectively call {\em untangling} and {\em tangling}. 

Given separable $T \in \F[x]$ of degree $d$ and a positive integer
$\mu$, polynomials modulo $T^\mu$ can naturally be written in the
power basis $1,x,\dots,x^{d\mu-1}$. Here we consider another representation,
based on bivariate polynomials. Introduce $\K:=\F[y]/\langle
T(y) \rangle$ with $\alpha$ the residue class of y; then, as an $\F$-algebra, $\F[x]/\langle T^\mu\rangle$
is isomorphic to $\K[\xi]/\langle \xi^\mu \rangle$ and untangling and
tangling are the corresponding change of bases that maps $x$ to
$\xi+\alpha$. Take, for instance,
$\F=\Q$, $T=x^2+x+2$ and $\mu=2$. Then $\K=\Q[y]/\langle
y^2+y+2\rangle$;
untangling is the isomorphism $ \Q[x]/\langle x^4+2x^3+5x^2+4x+4
\rangle \to \K[\xi]/\langle \xi^2\rangle$ and tangling is its inverse.

We now assume that $2,\dots,\mu-1$ are units in $\F$. Van der Hoeven
and Lecerf gave algorithms of quasi-linear cost for both untangling
and tangling; their algorithm for tangling is slightly slower than
that for untangling. Our first contribution is an improved algorithm
for tangling, using duality techniques inspired
by~\cite{Shoup94}. This saves logarithmic factors compared to the
results in~\cite{HoLe17}; it may be minor in practice, but we believe
this offers an interesting new point of view. Then we discuss how
these techniques can be of further use,  as in the resolution of
systems of the form $F(x_1,x_2,x_3)=G(x_1,x_2,x_3)=0$, for polynomials
$F,G$ in $\F[x_1,x_2,x_3]$. 

Our second main contribution is an extension of these algorithms to
situations involving more than one variable. As a first step, in this
paper, we deal with certain systems in two variables. 
Indeed, the discussion in~\cite{HoLe17} is closely related to the
question of how to describe isolated solutions of systems of
polynomial equations. This latter question has been the subject of
extensive work in the past; answers vary depending on what information
one is interested in.  

For the sake of this discussion, suppose we consider polynomials
$G_1,\dots,G_s$ in the variables $x_1$ and $x_2$, with coefficients in $\F$.  If
one simply wants to describe set-theoretically the (finitely many)
isolated solutions of $G_1,\dots,G_s$, popular choices include
description by means of univariate
polynomials~\cite{Macaulay16,Canny90,GiHe93,AlBeRoWo94,Rouillier99},
or triangular representations~\cite{Wu86,AuLaMo99}. When all isolated solutions
are non-singular nothing else is needed, but further questions arise
in the presence of multiple solutions as univariate or triangular 
representation may not be able to describe the local algebraic structure
at such roots.

The presence of singular isolated solutions means that the ideal
$\langle G_1,\dots,G_s\rangle$ admits a zero-dimensional primary
component that is not radical. Thus, let $I$ be a
zero-dimensional primary ideal in $\F[x_1,x_2]$ with radical
$\frakm$; we will suppose that $\F[x_1,x_2]/\frakm$ is separable
(which is always the case if $\F$ is perfect, for instance) to
prevent $\frakm$ from acquiring multiple roots over
an algebraic closure $\overline \F$ of $\F$.

A direct approach to describing the solutions of $I$, together with the
algebraic nature of $I$ itself, is to give one of its Gr\"obner
bases. Following~\cite{MaMoMo96}, one may also give a basis of the
dual of $\F[x_1,x_2]/I$, or a standard basis of
$I$. In~\cite[Section~5]{MaMoMo96}, Marinari, M\"oller and Mora make
the following interesting suggestion: build the field
$\K:=\F[y_1,y_2]/\tilde \frakm$, where $\tilde \frakm$ is the ideal
$\frakm$ with variables renamed $y_1,y_2$. Then the polynomials in $I$
vanish at $\alpha:=(\alpha_1,\alpha_2)$ when $\alpha_1,\alpha_2$ are the residue classes of
$y_1,y_2$ in $\K$.  Now extend $I$ to the
polynomial ring $\K[\xi_1,\xi_2]$, for new variables $\xi_1,\xi_2$, by
mapping $(x_1,x_2)$ to $(\xi_1,\xi_2)$. Then, the local structure of
$I$ at $\alpha$ can be described by the primary component of this
extended ideal at $\alpha$.

Let us show the similarities of this idea with van der Hoeven and
Lecerf's approach, on an example from~\cite{NeRaSc17}. We take
$\F=\Q$, $\frakm$ to be the maximal ideal $\langle T_1,T_2\rangle$, with
$T_1:= x_1^2+x_1+2$, $T_2:=x_2-x_1-1$, and $I=\frakm^2$ to be the $\frakm$-primary
ideal with generators
\vspace{-0.4em}
\begin{align*}
G_3 &= x_2^2 - 2x_1x_2 - 2x_2 + x_1^2 + 2x_1 + 1,\\[-1mm]
G_2 &= x_1^2x_2 + x_1x_2 + 2x_2 - x_1^3 - 2x_1^2 - 3x_1 - 2,\\[-1mm]
G_1 &= x_1^4 + 2x_1^3 + 5x_1^2 + 4x_1 + 4.\\[-7mm]
\end{align*}
 Since $T_2$ has degree
one in $x_2$, we can simply take $\K:=\Q[y_1]/\langle
y_1^2+y_1+2\rangle$, $\alpha_1$ to be the residue class of $y_1$ and
$\alpha_2 = \alpha_1 + 1$.

The $(\alpha_1,\alpha_2)$-primary component $J$ of the extension of
$I$ in $\K[\xi_1,\xi_2]$, i.e., the primary component associated to the 
prime ideal $(\xi_1-\alpha_1,\xi_2-\alpha_2)$, is the ideal with lexicographic
Gr\"obner basis
\begin{align*}
 H_3 &=    \xi_2^2 - 2\xi_2\alpha_1 - 2\xi_2 + \alpha_1 - 1,\\[-1mm]
 H_2 &=    \xi_1\xi_2 - \xi_2\alpha_1 - \xi_1\alpha_1 - \xi_1 - 2,\\[-1mm]
 H_1 &=  \xi_1^2 - 2\xi_1\alpha_1 - \alpha_1 - 2.\\[-7mm]
\end{align*}
Its structure appears more clearly after applying the translation
$(\xi_1,\xi_2) \mapsto (\xi_1+\alpha_1,\xi_2+\alpha_2)$: the
translated ideal $J'$ admits the very simple Gr\"obner basis $\langle
\xi_1^2, \xi_1\xi_2,\xi_2^2\rangle$. In other words, this
representation allows one to complement the set-theoretic description of
the solutions by the multiplicity structure.

Our first result in bivariate settings is the relation between the
Gr\"obner bases of $I$ and $J$ (or $J'$): in our example, they both
have three polynomials, and their leading terms are related by the
transformation $(\xi_1,\xi_2) \mapsto (x_1^2,x_2)$. We then prove
that, as in the univariate case, there is an $\F$-algebra isomorphism
$\F[x_1,x_2]/ I \to \K[\xi_1,\xi_2]/ J'$ given by $(x_1,x_2) \mapsto
(\xi_1+\alpha_1,\xi_2 + \alpha_2)$. In our example, this means that
$\Q[x_1,x_2]/ \langle G_1,G_2,G_3\rangle$ is isomorphic to
$\K[\xi_1,\xi_2]/\langle \xi_1^2, \xi_1\xi_2,\xi_2^2\rangle$.

Under certain assumptions on $J'$, we give algorithms for this
isomorphism and its inverse that extend those for univariate
polynomials; while their runtimes are not always quasi-linear, they are
subquadratic in the degree of $I$ (that is, the dimension of
$\F[x_1,x_2]/ I$). We end with a first application: upper bounds
on the cost of arithmetic operations in an algebra such as
$\F[x_1,x_2]/I$; these are new, to the best of our knowledge. Note that with a strong regularity assumption and in a different setting, it has been shown in~\cite{van2018fast} that multiplication in $\F[x1, x2] / I$ can be done in quasi-linear time.

Although our results are still partial (we make assumptions and deal only with bivariate systems), we believe it is
worthwhile to investigate these questions. In future work, we plan to
examine the impact of these techniques on issues arising from
polynomial system solving algorithms: a direction that one may
consider are lifting techniques in the presence of multiplicities, as
in~\cite{HaMoSz16} for instance, as well as the computation of GCDs
modulo ideals such as $I$ above. See, for instance,~\cite{Dahan2017} for
a discussion of the latter question.

\vspace{-0.5em}\section{Preliminaries}\label{ssec:prelim}

In the rest of this paper, $\F$ is a {\em perfect} field. The costs of
all our algorithms are measured in number of operations
$(+,-,\times,\div)$ in $\F$.

\paragraph{}\label{par:M}
We let $\M :\N \to \N$ be  such that product of elements
of degree less than $n$ in $\F[x]$ can be computed in $\M(n)$
operations, and such that $\M$ satisfies the super-linearity
properties of~\cite[Chapter~8]{GaGe13}.
Below, we will freely use all usual consequences of fast
multiplication (on fast GCD, Newton iteration, \dots) and refer the
reader to e.g.~\cite{GaGe13} for details. In particular,
multiplication in an $\F$-algebra of the form $\A:=\F[x]/\langle T(x)
\rangle$ with $T$ monic in $x$, or \sloppy $\A:=\F[x_1,x_2]/\langle
T_1(x_1),T_2(x_1,x_2) \rangle$ with $T_1$ monic in $x_1$ and $T_2$ monic
in $x_2$, can be done in time $O(\M(\delta))$, with $\delta:=\dim_\F
(\A)$.  Inversion, when possible, is slower by a logarithmic factor.
For $\A=\F[x_1,x_2]/I$, for a zero-dimensional monomial ideal $I$,
multiplication and inversion in $\A$ can be done in time
$O(\M(\delta)\log(\delta))$, resp.\ $O(\M(\delta)\log(\delta)^2)$,
with $\delta=\dim_\F (\A)$ (see the appendix).

\paragraph{}
We will use the {\em transposition
  principle}~\cite{CaKaYa89,KaKiBs88}, which is an algorithmic theorem
stating that if the $\F$-linear map encoded by an $n \times m$ matrix
over $\F$ can be computed in time $T$, the transposed map can be
computed in time $T + O(n+m)$. This result has been used in a variety
of contexts; our main sources of inspiration 
are~\cite{Shoup94,BoLeSc03}.

\paragraph{}\label{par:dual}
If $\A$ is an $\F$-vector space, its dual $\A^* := {\rm Hom}_\F(\A,\F)$
is the $\F$-vector space of $\F$-linear mappings $\A \to \F$. When
$\A$ is an $\F$-algebra, $\A^*$ becomes an $\A$-module: to a linear
mapping $\ell: \A \to \F$ and $F \in \A$ we can associate the linear
mapping $F\cdot \ell: G \in \A \mapsto \ell(F G)$. This operation is
called the {\em transposed product} in $\A^*$, since it is the transpose
of the multiplication-by-$F$ mapping. 

Given a basis $\mathcal{B}$ of $\A$, elements of $\A^*$ are
represented on the dual basis, by  their values
on $\mathcal{B}$. In terms of complexity, if $\A$ is an algebra such
as those in~\pref{par:M}, the transposition principle implies that
transposed products can be done in time $O(\M(\delta))$,
resp.\ $O(\M(\delta)\log(\delta))$, with again
$\delta:=\dim_\F(\A)$. See~\cite{Shoup99} for detailed algorithms in
the cases $\A=\F[x]/\langle T(x) \rangle$ and $\A=\F[x_1,x_2]/\langle
T_1(x_1),T_2(x_1,x_2) \rangle$.

An element $\ell\in \A^*$ is called a {\em generator} of $\A^*$ if $\A \cdot \ell
= \A^*$ (in other words, for any $\ell'$ in $\A^*$ there exists $F
\in \A$, which must be unique, such that $F \cdot \ell =
\ell'$). When $\A=\F[x]/\langle T(x) \rangle$, with $n:=\deg(T)$,
$\ell$ defined by $\ell(1)=\cdots=\ell(x^{n-2})=0$ and
$\ell(x^{n-1})=1$ is known to generate $\A^*$. For
$\A=\F[x_1,x_2]/\langle T_1(x_1),T_2(x_1,x_2) \rangle$, $\ell$ given
by $\ell(x_1^{n_1-1}x_2^{n_2-1})=1$, with all other $\ell(x_1^i
x_2^j)=0$, is a generator (here, we write $n_1:=\deg(T_1,x_1)$ and
$n_2:=\deg(T_2,x_2)$). For more general $\A$, $\A^*$ may not be
free: see for example Subsection~\ref{ssec:dualTang2}.


\vspace{-0.5em}\section{The univariate case revisited}\label{sec:univariate}
In this section, we work with univariate polynomials. Suppose that $T \in \F[x]$ is monic and separable (that is, without repeated roots in $ \overline \F$) with degree $d$, and let $\mu$ be an integer positive. We start from the following hypothesis:
\begin{description}
\item[${\bf H}_1$.]   $\F$ has characteristic at least $\mu$.
\end{description}
Define $\K:=\F[y]/T(y)$,
and let $\alpha$ be the residue class of $y$ in $\K$. Van der Hoeven
and Lecerf proved that the $\F$-algebra mapping
\[\begin{array}{cccc}
 \pi_{T,\mu} :& \F[x]/\langle T^\mu \rangle& \to     & \K[\xi]/\langle \xi^\mu \rangle\\
 &              x & \mapsto & \xi+\alpha
\end{array}\]
is well-defined and realizes an isomorphism of $\F$-algebras.
The mapping $\pi_{T,\mu}$ is called {\em untangling}, and its inverse ${\pi_{T,\mu}}^{-1}$ {\em tangling}.
Note that $\pi_{T,\mu}(F)$ simply computes the first $\mu$ terms of the Taylor expansion
of $F$ at $\alpha$, that is, $\pi_{T,\mu}(F) = \sum_{0 \le i < \mu} F^{(i)}(\alpha) \xi^i/i!$.

Reference~\cite{HoLe17} gives algorithms for both
untangling and tangling, the latter calling the former recursively;
the untangling algorithm runs in $O(\M(d\mu)\log(\mu))$ operations
in $\F$, while the tangling algorithm takes $O(\M(d\mu)\log(\mu)^2 +
\M(d) \log(d))$ operations. Using transposition techniques 
from~\cite{Shoup94}, we prove the following.
\begin{proposition}\label{prop:1}
  Given $G$ in $\K[\xi]/\langle \xi^\mu \rangle$, one can compute
  ${\pi_{T,\mu}}^{-1}(G)$ in $O(\M(d\mu)\log(\mu) + \M(d) \log(d))$ operations in $\F$.
\end{proposition}
\noindent The $\F$-algebra $\K$ admits the basis
$(1,\dots,\alpha^{d-1})$;  $\F[x]/\langle T^\mu\rangle$ has
 basis $\mathcal{B}=(1,x,\dots,x^{d \mu -1})$ and
$\K[\xi]/\langle \xi^\mu\rangle$ admits the bivariate basis
$\mathcal{C}=(1,\dots,\alpha^{d-1},\xi,\dots,\alpha^{d-1}\xi,\dots\xi^{\mu-1},\dots,\alpha^{d-1}\xi^{\mu-1}).$
As per~\pref{par:dual}, we represent a linear form
$L \in \F[x]/\langle T^\mu\rangle^*$ by the vector 
$[L(x^i) \mid 0 \le i < d \mu]$ $\in$ $\F^{d \mu}$,
and a linear form $\ell \in \K[\xi]/\langle \xi^\mu\rangle^*$ 
by the bidimensional vector
$[\ell(\alpha^i \xi^j) \ \mid \  0 \le i < d,\ 0 \le j <  \mu] \in \F^{d \times \mu}$.


\vspace{-0.9em}
\subsection{A faster tangling algorithm}

This section shows that using the transpose of untangling allows us to
deduce an algorithm for tangling; see~\cite{Shoup94,DeSc12} for a similar use of
transposition techniques. We start by describing useful subroutines.

\paragraph{} \label{par:inverse} The first algorithmic result we will
need concerns the cost of inversion in $\F[x]/\langle T^\mu
\rangle$. To compute $1/F \bmod T^\mu$ for some $F \in \F[x]$ of
degree less than $d\mu$ we may start by computing $\bar G := 1/ \bar F
\bmod T$, with $\bar F := F \bmod T$; this costs $O(\M(d\mu) +
\M(d)\log (d))$ operations in $\F$. Then we lift $\bar G$ to $G:=1/F
\bmod T^\mu$ by Newton iteration modulo the powers of $T$, at the cost
of another $O(\M(d \mu))$.

\paragraph{} \label{par:hank}
Next, we discuss the solution of certain Hankel systems.
Consider $L$ and $L'$, two $\F$-linear forms $\F[x]/\langle T^\mu
\rangle \to \F$; our goal is to find $F$ in $\F[x]/\langle T^\mu
\rangle$ such that $F\cdot L = L'$, under the assumption that $L$ generates 
the dual space $\F[x]/\langle T^\mu
\rangle^*$.
 In matrix terms, this is equivalent to finding
coefficients $f_0,\dots,f_{d\mu-1}$ of $F$ such that $[H][f_0,\dots,f_{d\mu-1}]^T = [B]$ with $H_{i,j} = L(x^{i+j})$ and $B_{i} = L'(x^{i})$, $0\leq i < d\mu$.
The system can be solved in \sloppy $O(\M(d \mu) \log (d\mu))$
operations in $\F$~\cite{BrGuYu80}, but we will derive an improvement
from the fact that $T^\mu$ is a $\mu$th power.

An algorithm that realizes the transposed product $(L,F) \mapsto
L'$ is in~\cite[Lemma~2.5]{BoJeMoSc17}: let $\zeta: \F^{d\mu} \to
\F^{d\mu}$ be the upper triangular Hankel operator with first column
the coefficients of degree $1,\dots,d\mu$ of $T^\mu$, and let
$\Lambda$ and $\Lambda'$ be the two polynomials in $\F[x]$ with respective
coefficients $\zeta(L)$ and $\zeta(L')$. Then $\Lambda' = F
\,\Lambda \bmod T^\mu$.

Given the values of $L$ and $L'$ at $1,\dots,x^{d\mu-1}$, we 
compute $\zeta(L)$ and $\zeta(L')$ in $O(\M(d\mu))$ operations. Since
$L$ generates $\F[x]/\langle T^\mu \rangle^*$, $\Lambda$ is invertible
modulo $T^\mu$; then, using \pref{par:inverse}, we compute its
inverse in $O(\M(d\mu) + \M(d)\log (d))$ operations. Multiplication by
$\Lambda'$ takes another $O(\M(d\mu))$ operations, for a total of
$O(\M(d\mu) + \M(d)\log (d))$. 

\paragraph{}\label{par1} 
We now recall van der Hoeven and Lecerf's algorithm for the mapping
$\pi_{T,\mu}$, and deduce an algorithm for its transpose, with the same
asymptotic runtime. Van der Hoeven and Lecerf's algorithm is
recursive, with a divide-and-conquer structure; the key idea is that
the coefficients of $\pi_{T,\mu}(F)$, for $F$ in $\F[x]/\langle T^\mu\rangle$,
are the values of $F,F',\dots,F^{(\mu-1)}$ at $\alpha$, divided
respectively by $0!,1!,\dots,(\mu-1)!$. 
\vspace{-0.7em}
\begin{algorithm}[H]
  \caption{$\pi_{\rm rec}(F, T, \mu)$}
  \begin{algorithmic}[1]
  \Require $F \in \F[x]/\langle T^\mu\rangle$
  \Ensure $[F(\alpha),\dots,F^{(\mu-1)}(\alpha)] \in \K^\mu$
    \IfThenElse {$\mu=1$}{\Return $[\,F(\alpha)\,]$}{set $\lambda := \lfloor \frac{\mu}{2} \rfloor$}
    \State  \Return $\pi_{\rm rec}(F \bmod T^\lambda, T, \lambda)$ cat $\pi_{\rm rec}(F^{(\lambda)} \bmod T^{\mu-\lambda},T,\mu-\lambda)$
  \end{algorithmic}
\end{algorithm}

\vspace{-2.2em}

\begin{algorithm}[H]
  \caption{$\pi(F, T, \mu)$}
  \begin{algorithmic}[1]
  \Require $F \in \F[x]/\langle T^\mu\rangle$ 
  \Ensure $\pi_{T,\mu}(F) \in \K[\xi]/\langle \xi^\mu \rangle$
  \State \Return $\sum_{0 \le i < \mu} \frac{v[i]}{i!} \xi^i$, with  $v:=\pi_{\rm rec}(F, T, \mu)$
  \end{algorithmic}
\end{algorithm}
The runtime $\mathcal{T}(d,\mu)$ of $\pi_{\rm rec}$ satisfies $\mathcal{T}(d,\mu) \le
\mathcal{T}(d,\mu/2)+ O(\M(d\mu))$, so this results in an algorithm for $\pi_{T,\mu}$
that takes $O(\M(d\mu)\log(\mu))$ operations.  Since $\pi_{T,\mu}$ is an
$\F$-linear mapping $\F[x]/\langle T^\mu\rangle \to \K[\xi]/\langle
\xi^\mu\rangle$, its transpose ${\pi_{T,\mu}}^\perp$ is an $\F$-linear mapping
$\K[\xi]/\langle \xi^\mu\rangle^* \to \F[x]/\langle
T^\mu\rangle^*$. The transposition principle implies that ${\pi_{T,\mu}}^\perp$ can
be computed in $O(\M(d\mu)\log(\mu))$ operations; we make the
corresponding algorithm explicit as follows.

We transpose all steps of the algorithm above, in reverse order.  As
input we take $\ell \in \K[\xi]/\langle \xi^\mu \rangle^*$, which we
see as a bidimensional vector in $\F^{d \times \mu}$; we also
write $\ell=[\ell_i \mid 0 \le i < \mu]$, with all $\ell_i$ in $\F^d$.
The transpose of the concatenation at the last step allows one
to apply the two recursive calls to the first and second halves of
input $\ell$. Each of them is followed by an application of the
transpose of Euclidean division (see below), and after
``transpose differentiating'' the second intermediate result (see
below), we return their sum.

\vspace{-1em}

\begin{algorithm}[H]
  \caption{$\pi_{\rm rec}^\perp (\ell, T, \mu)$}
    \begin{algorithmic}[1]
  \Require $\ell\in \F^{d \times \mu}$
    \IfThenElse{$\mu=1$}{\Return $\ell_0$}{$\lambda := \lfloor \frac{\mu}{2} \rfloor$}
    \State $v_0 := \pi_{\rm rec}^\perp([\ell_i \ \mid \ 0 \le i <\lambda],T,\lambda)$
    and $u_0 := \bmod^\perp(v_0, T^\lambda, d \mu)$
    \State $v_1 := \pi_{\rm rec}^\perp([\ell_i \ \mid \ \lambda \le i < \mu],T,\mu-\lambda)$
    \State $u_1 := \text{diff}^\perp( \bmod^\perp(v_1, T^{\mu-\lambda}, d \mu-\lambda)) ,\lambda)$
    \State \Return $u_0+u_1$
    \end{algorithmic}
\end{algorithm}
\vspace{-4.4em}
\begin{algorithm}[H]
  \caption{$\pi^\perp (\ell, T, \mu)$}
    \begin{algorithmic}[1]
  \Require $\ell \in \K[\xi]/\langle \xi^\mu \rangle^* \simeq \F^{d \times \mu}$
    \Ensure ${\pi_{T,\mu}}^\perp(\ell) \in \F[x]/\langle T^\mu\rangle^* \simeq \F^{d \mu}$ 
\State \Return $\pi_{\rm rec}^\perp ([\ell_i/i! \ \mid \ 0 \le i < \mu], T, \mu)$
    \end{algorithmic}
\end{algorithm} 
\vspace{-1.5em}
\noindent Correctness follows from the correctness of van der Hoeven and
Lecerf's algorithm. Following~\cite{BoLeSc03}, given a vector $u$, a polynomial $S \in \F[x]$ and
an integer $t \ge \deg(S)$, where $u$ has length $\deg(S)$,
$\bmod^\perp(u, S, t)$ returns the first $t$ terms of the sequence
defined by initial conditions $u$ and minimal polynomial $S$ in time $O(\M(t))$.
Given a vector $u$ of
length $t-\lambda$, $v:=\text{diff}^\perp(u,\lambda)$ is the vector of
length $t$ given by $v_0=\cdots=v_{\lambda-1}=0$ and
$v_i =i\cdots (i-\lambda+1) u_{i-\lambda}$ for $i=\lambda,\dots,t-1$.
It can be computed in linear time $O(t)$.
Overall, as in~\cite{HoLe17}, the runtime is $O(\M(d
\mu) \log(\mu))$.

\paragraph{} We can now give our algorithm for the tangling operator ${\pi_{T,\mu}}^{-1}$;
it is inspired by a similar result due to
Shoup~\cite{Shoup94}.

Take $G$ in $\K[\xi]/\langle \xi^\mu \rangle$: we want to find $F \in
\F[x]/\langle T^\mu\rangle$ such that ${\pi_{T,\mu}}(F) = G$. Let $\ell:
\K[\xi]/\langle \xi^\mu \rangle \to \F$ be defined by
$\ell(\alpha^{d-1}\xi^{\mu-1})=1$ and $\ell(\alpha^i\xi^j)=0$ for all
other values of $i <d,j<\mu$; as pointed out in~\pref{par:dual}, this
is a generator of $\K[\xi]/\langle \xi^\mu \rangle^*$.  Define further
$\ell':=G \cdot \ell$. Then $\ell'$ is a transposed product as
in~\pref{par:dual}, and we saw that it can be computed in
$O(\M(d\mu))$ operations. This  implies ${\pi_{T,\mu}}(F)\cdot \ell
= \ell'$.

Let now $L:={\pi_{T,\mu}}^\perp(\ell)$ and $L':={\pi_{T,\mu}}^\perp(\ell')$; we obtain them by
applying our transpose untangling algorithm to $\ell$, resp.\,
$\ell'$, in time $O(\M(d\mu)\log(\mu) + \M(d)\log (d))$.  Since $\ell$
is a generator of $\K[\xi]/\langle \xi^\mu \rangle^*$, $L$ is a
generator of $\F[x]/\langle T^\mu \rangle^*$. The equation
${\pi_{T,\mu}}(F)\cdot \ell = \ell'$ then implies that $F \cdot L = L'$, which
is an instance of the problem discussed in~\pref{par:hank};
applying the algorithm there takes another $O(\M(d\mu) +
\M(d)\log (d))$. Summing all costs, this gives an algorithm for
${\pi_{T,\mu}}^{-1}$ with cost $O(\M(d\mu)\log(\mu) + \M(d)\log (d))$,
proving Proposition~\ref{prop:1}.


\subsection{Applications}

\paragraph{}
For $P$ in $\F[x]$ one can compute $x^D \bmod
P$ using $O(\log(D))$ multiplications modulo $P$ by repeated
squaring. 
Applications include Fiduccia's algorithm for the computation of terms
in linearly recurrent sequences~\cite{Fiduccia85} or 
of high powers of matrices
\cite{ranum1911general,Giesbrecht1995}. This algorithm takes
$O(\M(n) \log(D))$ operations in $\F$, with $n:=\deg(P)$. We
assume without loss of generality that $D \ge n$.

We can do better, in cases where $P$ is not
squarefree. For computations of terms in recurrent sequences, such
$P$'s appear when computing terms of \textit{bivariate} recurrent
sequences $(a_{i,j})$ defined by $ \sum_{i,j} a_{i,j} x^i y^j=
N(x,y)/Q(x,y)$, for some polynomials $N,Q \in \F[x,y]$ with $Q(0,0)
\ne 0$. Then, the $j$-th row $\sum_{i} a_{i,j} x^i$ has
characteristic polynomial $P^j$, where $P$ is the reverse polynomial
of $Q(x,0)$~\cite{BoCaChDu}.

First, assume that $P=T^\mu$ with $T$ separable of degree $d$. Then
we compute $x^D \bmod P$ by tangling $r := (\xi + \alpha)^D$. The quantity $r = \sum_{i=0}^{\mu-1}
\binom{D}{i} \xi^i \alpha^{D-i}$ can be computed in time
$O(\M(d)(\log(D) + \mu ))$, by computing
$\alpha^{D-\mu+1},\alpha^{D-\mu+2},\dots,\alpha^D$ and multiplying
them by the binomial coefficients (which themselves are obtained by
using the recurrence they satisfy).
By Proposition~\ref{prop:1}, the cost of tangling is $O( \M(d\mu)
\log(\mu) + \M(d)\log(d))$, which brings the total  to
$O(\M(d)\log(D) + \M(d\mu) \log(\mu))$, since $d \le D$.  To
compute $x^D$ modulo an arbitrary $P$, one may compute the squarefree
decomposition of $P$, apply the previous algorithm modulo each factor
and obtain the result by applying the Chinese Remainder Theorem. The
overall runtime becomes $O(\M(m)\log(D) + \M(n) \log(n))$, where $n$
and $m$ are the degrees of $P$ and its squarefree part, respectively;
this is to be compared with the cost $O(\M(n) \log(D))$ of repeated
squaring.
While this  algorithm improves over the
direct approach, practical gains show up only for astronomical values of
the parameters.

\paragraph{}
Assume $\F=\Q$. In~\cite{LeMeSc13}, Lebreton, Mehrabi and Schost gave
an algorithm to compute the intersection of surfaces in 3d-space, that
is, to solve polynomial systems of the form $F(x_1,x_2,x_3) =
G(x_1,x_2,x_3) = 0$. Assuming that the ideal $K:=\langle F, G \rangle
\subset \Q(x_1)[x_2,x_3]$ is radical and that we are in generic
coordinates, the output is polynomials $S, T, U$ in $\Q[x_1,x_2]$ such
that $K$ is equal to $\langle S, U x_3-T \rangle$ (so $S$ describes
the projection of the common zeros of $F$ and $G$ on the
$x_1,x_2$-plane, and $T$ and $U$ allow us to recover $x_3$).  The
algorithm of~\cite{LeMeSc13} is Monte Carlo, with runtime $O(D^{4.7})$
where $D$ is an upper bound on  $\deg(F)$ and $\deg(G)$. The output
has $\Theta(D^4)$ terms in the worst case, and the result
in~\cite{LeMeSc13} is the best to date.

The case of non-radical systems was discussed in~\cite{MeSc16}. It was
pointed out in the introduction of that paper that quasi-linear time
algorithms for untangling and tangling (which were not explicitly
called by these names) would make it possible to extend the results
of~\cite{LeMeSc13} to general systems. Hence, already with the results 
by van der Hoeven and Lecerf a runtime $O(D^{4.7})$ was made possible
for the problem of surface intersection, without a radicality assumption.


\section{The bivariate case}

We now  generalize the previous questions to the bivariate
setting. We expect several of these ideas to carry over to higher
numbers of variables, but some adaptations may be non-trivial (for
instance, we rely on Lazard's structure theorem on lexicographic
bivariate Gr\"obner bases). As an application, we give results on the
complexity of arithmetic modulo certain primary ideals.


\subsection{Setup}

\paragraph{}\label{par:set}
For the rest of the paper, the {\em degree} $\deg(I)$ of a
zero-dimensional ideal $I$ in $\F[x_1,x_2]$ is defined as the
dimension of $\F[x_1,x_2]/I$ as a vector space (the same definition
will hold for polynomials over any field).

Let $\frakm$ be a maximal ideal of degree $d$ in $\F[x_1,x_2]$; we
consider two new variables $y_1,y_2$, we let $\gamma: \F[x_1,x_2] \to
\F[y_1,y_2]$ be the $\K$-algebra isomorphism mapping $(x_1,x_2)$ to
$(y_1,y_2)$ and let $\tilde \frakm:=\gamma(\frakm)$. This is a maximal
ideal as well, and $\K:=\F[y_1,y_2]/\tilde \frakm$ is a field
extension of degree $d$ of $\F$. We then let $\alpha_1,\alpha_2$ be the
respective residue classes of $y_1,y_2$ in~$\K$.

Next, let $J \subset \K[\xi_1,\xi_2]$, for two new variables
$\xi_1,\xi_2$, be a zero-dimensional primary ideal at
$\alpha:=(\alpha_1,\alpha_2)$. Finally, let $I:=\Phi^{-1}(J)$, where
$\Phi$ is the natural embedding $\F[x_1,x_2] \to \K[\xi_1,\xi_2]$
given by $(x_1,x_2) \mapsto (\xi_1,\xi_2)$. One easily checks that $I$
is $\frakm$-primary (that is, $\frakm$ is the radical of $I$), and
that $J$ is the primary component at $\alpha$ of the ideal $I\cdot
\K[\xi_1,\xi_2]$ generated by $\Phi(I)$. Note that since $\F$ is
perfect, $\F \to \K$ is separable, so
over an algebraic closure $\overline \F$ of $\F$, $\frakm$
has $d$ distinct solutions.
We make the following assumption:

\begin{description}
\item[${\bf H}_2$.] $\F$ has characteristic at least $n$, with $n:=\deg(I)$.
\end{description}

\noindent Finally, we let $J' \subset \K[\xi_1,\xi_2]$ be the ideal obtained by
applying the translation $(\xi_1,\xi_2) \mapsto (\xi_1 +
\alpha_1, \xi_2 + \alpha_2)$ to $J$; it is primary at $(0,0)$.

\paragraph{} \label{par:mI}
Although our construction starts from the datum of $\frakm$
and $J \subset \K[\xi_1,\xi_2]$ and defines $I$ from them, we may also
take as starting points $\frakm$ and an $\frakm$-primary ideal $I
\subset \F[x_1,x_2]$ (this is what we did for the example in the
introduction).

Under that point of view, consider the ideal $I\cdot \K[\xi_1,\xi_2]$
generated by $\Phi(I)$, for $\Phi: \F[x_1,x_2] \to \K[\xi_1,\xi_2]$ as
above, and let $J$ be the primary component of $I\cdot
\K[\xi_1,\xi_2]$ at $\alpha$. One verifies that $I$ is equal to
$\Phi^{-1}(J)$, so we are indeed in the same situation as
in~\pref{par:set}.

\paragraph{} \label{par:defGB}
For the rest of the paper, we use the lexicographic monomial
ordering in $\F[x_1,x_2]$ induced by $x_1 < x_2$, and its analogue in
$\K[\xi_1,\xi_2]$; ``the'' Gr\"obner basis of an ideal is its minimal
reduced Gr\"obner basis for this order. Our first goal in this section is then to
describe the relation between the Gr\"obner bases of $I$ and $J$:
viz., they have the same number of polynomials, and their leading
terms are related in a simple fashion (as seen on the example above).

Let $T$ be the Gr\"obner basis of $\frakm$. Since $\frakm$ is maximal,
$T$ consists of two polynomials $(T_1,T_2)$, with $T_1$ of degree $d_1$
in $\F[x_1]$ and $T_2$ in $\F[x_1,x_2]$, monic of degree $d_2$ in
$x_2$. Note that $d_1d_2 = d = \deg(\frakm)$.
Next, let $H = (H_1,\dots,H_t)$ be the Gr\"obner basis of $J$, with
$H_1 < \cdots < H_t$; we let
$\xi_1^{\mu_1}\xi_2^{\nu_1},\dots,\xi_1^{\mu_t}\xi_2^{\nu_t}$ be the
respective leading terms of $H_1,\dots,H_t$. Thus, the $\mu_i$'s are
decreasing, the $\nu_i$'s are increasing, and $\nu_1=\mu_t=0$.
Finally, we let $\mu:=\deg(J) = \deg(J')$. Remark that the Gr\"obner
basis of $J'$ admits the same leading terms as $H$.

In our example, we have $t=3$, $(\mu_1,\nu_1)=(2,0)$,
$(\mu_2,\nu_2)=(1,1)$ and $(\mu_3,\nu_3)=(0,2)$. The integers $d_1,d_2$
are respectively $2$ and $1$, so $d=2$, the degree $n$ is $6$ and the
multiplicity $\mu$ is $3$. The key result in this subsection is
the following.
\begin{proposition}\label{prop:2}
  The Gr\"obner basis of $I$ has the form $(R_1,\dots,R_t)$, where for
  $j=1,\dots,t$, $R_j = {T_1}^{\mu_j} \tilde R_j$, for some polynomial
  $\tilde R_j \in \F[x_1,x_2]$ monic of degree $d_2 \nu_j$ in $x_2$.
 In particular, $n = d \mu$.
\end{proposition}
As a result, for all $j$ the leading term of $R_j$ is ${x_1}^{d_1
  \mu_j} {x_2}^{d_2 \nu_j}$, whereas that of $H_j$ is ${\xi_1}^{\mu_j}
{\xi_2}^{\nu_j}$, as in our example. The next two
sub-sections are devoted to the proof of this proposition.

\paragraph{}  We define here a family
of polynomials $G_1,\dots,G_t$, and prove that they form a
(non-reduced) Gr\"obner basis of $I$ in~\pref{par:GB}.

Because the extension $\F \to \K$ is separable, it admits a primitive
element $\beta$, with minimal polynomial $F \in \F[t]$; this
polynomial has degree $[\K:\F]=d$.  Let $\L$ be a splitting field
for $F$ containing $\K$ and let $I\cdot \L[\xi_1,\xi_2]$ and $K$ be
the extensions of $I\cdot\K[\xi_1,\xi_2]$ and $J$ in
$\L[\xi_1,\xi_2]$, respectively. Then $\deg(J) = \deg(K)$, and $K$ is the
primary component of $I\cdot \L[\xi_1,\xi_2]$ at $\alpha$.

Let $\beta_1=\beta,\beta_2,\dots,\beta_d$ be the roots of $F$ in $\L$.
For all $i=1,\dots,d$, we let $\sigma_i$ be an element in the Galois
group of $\L/\F$ such that $\beta_i=\sigma_i(\beta)$, as well as
$\alpha^{(i)}:=(\sigma_i(\alpha_1),\sigma_i(\alpha_2))$. Note that
these elements are pairwise distinct: since $\beta$ is in
$\F[\alpha_1,\alpha_2]$ and all $\sigma_i$'s fix $\F$,
$\alpha^{(i)}=\alpha^{(j)}$ implies $\beta_i=\beta_j$, and thus
$i=j$. Therefore, $\alpha^{(1)},\dots, \alpha^{(d)}$ can be seen as all the
roots of $\frakm$, with $\alpha^{(1)}=\alpha$.

For $i =1,\dots,d$, let $K_i$ be the primary component of $I \cdot
\L[\xi_1,\xi_2]$ at $\alpha^{(i)}$, so that $K_1 = K$. By
construction, these ideals are pairwise coprime, and their product is
$I \cdot \L[\xi_1,\xi_2]$.  Take $i$ in $1,\dots,d$, and let $D$ be a
large enough integer such that $K = I \cdot \L[\xi_1,\xi_2] +
\frakn^D$ and $K_i = I \cdot \L[\xi_1,\xi_2] + \frakn_i^D$, with
$\frakn$ and $\frakn_i$ the maximal ideals at $\alpha$ and
$\alpha^{(i)}$ respectively. Since $I \cdot \L[\xi_1,\xi_2]$ is defined over $\F$,
$\sigma_i$ thus maps the generators of $K$ to those of $K_i$. This
implies that the Gr\"obner basis of $K_i$ is
$(H_{i,1},\dots,H_{i,t})$, with $H_{i,j} := \sigma_i(H_{j})$ for all
$j \le t$.

By definition of the integers $d_1,d_2$, we can partition the roots
$\{\alpha^{(1)},\dots,\alpha^{(d)}\}$ of $\frakm$ according to their
first coordinate, into $d_1$ classes $C_1,\dots,C_{d_1}$ of cardinality $d_2$
each: for $\kappa \le d_1$, all $\alpha^{(i)}$ in $C_\kappa$ have the
same first coordinate, say $\zeta_\kappa$, and the $\zeta_\kappa$'s
are pairwise distinct. Remark that $\zeta_1,\dots,\zeta_{d_1}$ are the
roots of $T_1$.
  
Fix $\kappa \le d_1$ and take $i$ such that $\alpha^{(i)}$ is in
$C_\kappa$. Because $K_i$ is primary at $\alpha$, Lazard's structure
theorem on bivariate lexicographic Gr\"obner bases~\cite{Lazard85}
implies that for $j=1,\dots,t$, $H_{i,j} =
(\xi_1-\zeta_\kappa)^{\mu_j} \tilde H_{i,j}$, for some polynomial
$\tilde H_{i,j} \in \L[\xi_1,\xi_2]$, monic of degree $\nu_j$ in
$\xi_2$, and of degree less than $\mu_1-\mu_j$ in $\xi_1$.

For $1 \le \kappa \le d_1$ and $1\le j \le t$, let us then define
$\tilde G_{\kappa,j} := \prod_i \tilde H_{i,j}$, where the product is
taken over all $i$ such that $\alpha^{(i)} \in C_\kappa$. This is a
polynomial in $\L[\xi_1,\xi_2]$, with leading term ${\xi_2}^{d_2
  \nu_j}$. Finally, let $\tilde G_1 := 1$, and for $2 \le j \le t$ let
$\tilde G_j$ be the unique polynomial in $\L[\xi_1,\xi_2]$ of degree
less than $d_1(\mu_{1}-\mu_j)$ in $\xi_1$ such that $\tilde G_j \bmod
(\xi_1-\zeta_\kappa)^{\mu_{1}-\mu_j} =\tilde G_{\kappa,j}$ holds for all
$\kappa \le d_1$. We claim that $(G_1,\dots,G_t)$, with $G_j :=
{T_1}^{\mu_j} \tilde G_j$ for all $j$, is a Gr\"obner basis of $I
\cdot \L[\xi_1,\xi_2]$, minimal but not necessarily reduced.
 
\paragraph{}\label{par:GB} To establish this claim, we first prove that $I \cdot \L[\xi_1,\xi_2]=\langle G_1,\dots,G_t \rangle$ in
$\L[\xi_1,\xi_2]$. The first step is to determine the common zeros of
$G_1,\dots,G_t$.  Since $G_1= {T_1}^{\mu_1}$, the $\xi_1$-coordinates
of the solutions are the roots $\{\zeta_1,\dots,\zeta_{d_1}\}$ of
$T_1$. Fix $\kappa \le d_1$, and let $(\zeta_\kappa,\eta)$ be a root
of $G_1,\dots,G_t$. In particular, $G_t(\zeta_\kappa, \eta)=\tilde
G_t(\zeta_\kappa, \eta)=0$. This implies that $\tilde
G_{\kappa,t}(\zeta_\kappa, \eta)=0$, so there exists $i \le d$ such
that $(\zeta_\kappa,\eta)=\alpha^{(i)}$. Conversely, any
$\alpha^{(i)}$ cancels $G_1,\dots,G_t$, so that the zero-sets of
$G_1,\dots,G_t$ and $I \cdot \L[\xi_1,\xi_2]$ are equal.  Next, we
determine the primary component $Q_i$ of $\langle G_1,\dots,G_t
\rangle$ at a given $\alpha^{(i)}$.

Take such an index $i$, and assume that $\alpha^{(i)}$ is in
$C_\kappa$, for some $\kappa\le d_1$ (so the first coordinate of
$\alpha^{(i)}$ is $\zeta_\kappa$). Take $D$ large enough, so that
$D\ge \mu_1$ and $(\xi_1-\zeta_\kappa)^D$ belongs to $Q_i$; hence
$Q_i$ is also the primary component of the ideal $\langle
G_1,\dots,G_t, (\xi_1-\zeta_\kappa)^D\rangle$ at $\alpha^{(i)}$. This
ideal is generated by the polynomials $(\xi_1-\zeta_\kappa)^{\mu_1}$
and $(\xi_1-\zeta_\kappa)^{\mu_j} \tilde G_j$, for $2\le j \le t$.
For such $j$, since $\tilde G_j \bmod
(\xi_1-\zeta_\kappa)^{\mu_1-\mu_j} = \tilde G_{\kappa,j}$, we get
that $(\xi_1-\zeta_\kappa)^{\mu_j} \tilde G_j \bmod
(\xi_1-\zeta_\kappa)^{\mu_1} = (\xi_1-\zeta_\kappa)^{\mu_j}\tilde
G_{\kappa,j}$. As a result, the ideal above also admits the generators
$(\xi_1-\zeta_\kappa)^{\mu_1},(\xi_1-\zeta_\kappa)^{\mu_2} \tilde
G_{\kappa,2},\dots,\tilde G_{\kappa,t}$. Now, recall that $\tilde
G_{\kappa,j}= \prod_{\iota} \tilde H_{\iota,j}$, where the product is
taken over all $\iota$ such that $\alpha^{(\iota)}$ is in
$C_\kappa$. For $\iota \ne i$, $\tilde H_{\iota,j}$ does not vanish at
$\alpha^{(i)}$~\cite[Th.~2.(i)]{Lazard85}, so it is invertible locally
at $\alpha^{(i)}$. It follows that the primary component of $G$ at
$\alpha^{(i)}$ is generated by
$(\xi_1-\zeta_\kappa)^{\mu_1},(\xi_1-\zeta_\kappa)^{\mu_2} \tilde
H_{i,2},\dots,\tilde H_{i,t}$, that is, $H_{i,1},\dots,H_{i,t}$.  This
is precisely the ideal $K_i$.

To summarize, $\langle G_1,\dots,G_t\rangle$ and $I \cdot
\L[\xi_1,\xi_2]$ have the same primary components $K_1,\dots,K_d$, so
these ideals coincide. It remains to prove that $(G_1,\dots,G_t)$ is a
Gr\"obner basis of $I \cdot \L[\xi_1,\xi_2]$. The shape of the leading
terms of $G_1,\dots,G_t$ implies that number of monomials reduced with
respect to these polynomials is $d\deg(J) =d \mu$. Now, since all its
primary components $K_i$ have degree $\mu=\deg(J)$, the ideal $I \cdot
\L[\xi_1,\xi_2]=\langle G_1,\dots,G_t \rangle$ has degree $d \mu$ as
well.  As a result, $G_1,\dots,G_t$ form a Gr\"obner basis (since
otherwise, applying the Buchberger algorithm to them would yield fewer
reduced monomials, a contradiction).

The polynomials $G_1,\dots,G_t$ are a Gr\"obner basis, {\em minimal},
as can be seen from their leading terms, but not reduced; we
let $R_1,\dots,R_t$ be the corresponding reduced minimal Gr\"obner
basis. For all $j$, ${T_1}^{\mu_j}$ divides $G_j$, and we obtain $R_j$
by reducing $G_j$ by multiples of ${T_1}^{\mu_j}$, so that each $R_j$
is a multiple of ${T_1}^{\mu_j}$ as well. In addition, the leading
terms of $G_j$ and $R_j$ are the same. Hence, our proposition is
proved.

\paragraph{} As a corollary, the following proposition and its proof
extend~\cite[Lemma~9]{HoLe17} to bivariate contexts. We will still use
the names {\em untangling} and {\em tangling} for $\pi_{\frakm,J'}$ as
defined below and its inverse.
\begin{proposition}\label{prop:pi2}
  Assume $\frakm$ is a maximal ideal in $\F[x_1,x_2]$ and $I$ is an
  $\frakm$-primary zero-dimensional ideal in $\F[x_1,x_2]$, with $\F$
  perfect of characteristic at least $\deg(I)$. 

Let $\tilde \frakm$ be the
  image of $\frakm$ through the isomorphism $\F[x_1,x_2] \simeq
  \F[y_1,y_2]$, let $\alpha_1,\alpha_2$ be the residue classes of
  $y_1,y_2$ in $\K:=\F[y_1,y_2]/\tilde\frakm$ and let $J$ be the
  primary component of $I \cdot \K[\xi_1,\xi_2]$ at
  $(\alpha_1,\alpha_2)$. Finally, let $J'$ be the image of $J$ through
  $(\xi_1,\xi_2) \mapsto (\xi_1+\alpha_1,\xi_2+\alpha_2)$.  Then,
  there exists an $\F$-algebra isomorphism
  \begin{align}
    \pi_{\frakm,J'}: \F[x_1,x_2]/I \to \K[\xi_1,\xi_2]/J'
  \end{align}
  given by $(x_1,x_2) \mapsto (\xi_1 + \alpha_1, \xi_2 + \alpha_2)$
\end{proposition}

\begin{proof}
  We prove that the embedding $\Phi:\F[x_1,x_2] \to \K[\xi_1,\xi_2]$
  given by $(x_1,x_2) \mapsto (\xi_1,\xi_2)$ induces an isomorphism of
  $\F$-algebras $\F[x_1,x_2]/I \to \K[\xi_1,\xi_2]/J$. From
  this, applying the change of variables $(\xi_1,\xi_2) \mapsto (\xi_1 +
  \alpha_1, \xi_2 + \alpha_2)$ gives the result.

  Since $\Phi(I)$ is contained in $J$, the embedding $\Phi$ induces an
  homomorphism $ \phi: \F[x_1,x_2]/I \to \K[\xi_1,\xi_2]/J.$ By the
  previous proposition, both sides have dimension $d \mu$ over
  $\F$, so it is enough to prove that $\phi$ is injective. But this
  amounts to verifying that $\Phi^{-1}(J) = I$, which is true by definition.
\end{proof}


\subsection{Untangling for monomial ideals}

\paragraph{}
In this section, we give an algorithm for the mapping
$\pi_{\frakm,J'}$ of Proposition~\ref{prop:pi2} under a simplifying
assumption. To state it, recall that $J'$ is maximal at $(0,0) \in
\K^2$.  Then, our assumption is

\begin{description}
\item[${\bf H}_3$.] $J'$ is a {\em monomial} ideal. 
\end{description}
In view of the shape of the leading terms given in~\pref{par:defGB}
for the ideal $J$, we deduce that $J'=\langle
\xi_1^{\mu_1},\xi_1^{\mu_2} \xi_2^{\nu_2}, \dots,\xi_2^{\nu_t}
\rangle$.  In the rest of this subsection, $\mathcal{B}$ is the
monomial basis of $\F[x_1,x_2]/I$ induced by the Gr\"obner basis
exhibited in Proposition~\ref{prop:2} and $\mathcal{B}'$ is the
monomial basis of $\K[\xi_1,\xi_2]/J'$. Then, the inputs of the
algorithms in this subsection are in ${\rm Span}_\F
\mathcal{B}:=\oplus_{b\in \mathcal{B}} \F b$, and the outputs in ${\rm
  Span}_\K \mathcal{B}':=\oplus_{b'\in \mathcal{B}'} \K b'$. This
being said, our result is the following.

\begin{proposition}\label{prop:pi}
  Under  ${\bf H}_2$ and ${\bf H}_3$, given $F$ in
  $\F[x_1,x_2]/I$ one can compute $\pi_{\frakm,J'}(F)$ using either $O(\M(d n))$
  or $O(\M(\mu n) \log(\mu))$ operations in $\F$, and 
  in particular in $O(\M(n^{1.5}) \log(n))$ operations.
\end{proposition}
We prove the first two bounds in~\pref{par:algo1}
and~\pref{par:algo2} respectively. The last statement readily follows, since 
$n =d \mu$ (Proposition~\ref{prop:2}).

\paragraph{} \label{par:algo1}We start with an efficient algorithm for those cases 
where $d=[\K:\F]$ is small. The idea is simple: as in the
univariate case, the untangling mapping $\pi_{\frakm,J'}$ can be rephrased in
terms of Taylor expansion. Explicitly, for $F$ in $\F[x_1,x_2]/I$,
$\pi_{\frakm,J'}(F)$ is simply 
$$F(\xi_1+\alpha_1,\xi_2+\alpha_2) \bmod \langle
\xi_1^{\mu_1},\xi_1^{\mu_2} \xi_2^{\nu_2}, \dots,\xi_2^{\nu_t}
\rangle.$$ We compute $F(\xi_1+\alpha_1,\xi_2+\alpha_2)$, proceeding
one variable at a time.

\smallskip\noindent{{\em Step 1.}~~} Compute
$F^*:=F(\xi_1+\alpha_1,\xi_2) \in \K[\xi_1,\xi_2]$.  Because
$2,\dots,n$ are units in $\F$, given a univariate polynomial $P$ of
degree $t \le n$ in $\K[\xi_1]$ one can compute $P(\xi_1+\alpha_1)$
in $O(\M(t))$ operations $(+,\times)$ in $\K$ (see~\cite{AhStUl75}).
Using Kronecker substitution~\cite[Chapter~8.4]{GaGe13}, this
translates to $O(\M(d t))$ operations in $\F$ (we will systematically
use such techniques, see e.g. Lemma~2.2 in~\cite{GaSh92} for details).
Computing $F^*$ is done by
applying this procedure coefficient-wise with respect to $\xi_2$; in
particular, all $\xi_1$-degrees involved are at most $n$, and add up
to $n$.  The super-linearity of $\M$ implies that this takes a total
of $O(\M(d n))$ operations in $\F$.

\smallskip\noindent{{\em Step 2.}~~} Compute
$F^*(\xi_1,\xi_2+\alpha_2)=F(\xi_1+\alpha_1,\xi_2+\alpha_2)$.  This is
done in the same manner, applying the translation with respect to
$\xi_2$ instead; the runtime is still $O(\M(d n))$ operations in $\F$.

\smallskip\noindent{{\em Step 3.}~~} Since $F$ is in ${\rm Span}_\F
\mathcal{B}$, and $\mathcal{B}$ is stable by division, 
$F(\xi_1+\alpha_1, \xi_2+\alpha_2)$ are in ${\rm Span}_\K
\mathcal{B}:=\oplus_{b\in \mathcal{B}} \K b$. By
Proposition~\ref{prop:2}, all monomials in $\mathcal{B}'$ are in
$\mathcal{B}$, so we can obtain $\pi_{\frakm,J'}(F)$ by discarding
from $F(\xi_1+\alpha_1,\xi_2+\alpha_2)$ all monomials not in
$\mathcal{B}'$.

Overall, the runtime is $O(\M(d n))$ operations in $\F$. For small $d$,
when the multiplicity $\mu$ is large, this is close to being
linear in $n=\deg(I)$.

\paragraph{} \label{par:algo2} Next we give an another solution, which will perform well in
cases where the multiplicity $\mu = \deg(J')$ is small.

Again the idea is simple: given $F$ in ${\rm Span}_\F \mathcal{B}$,
compute $F(\xi_1+\alpha_1,\xi_2+\alpha_2) \bmod \langle
{\xi_1}^{\mu_1}, {\xi_2}^{\nu_t}\rangle$, and again discard unwanted
terms (this is correct, since all coefficients of $\pi_{\frakm,J'}(F)$
are among those we compute). As in the previous paragraph, this is
done one variable at a time; in the following, recall that $\frakm =
\langle T_1(x_1), T_2(x_1,x_2) \rangle$, with $\deg(T_1,x_1)=d_1$ and
$\deg(T_2,x_2)=d_2$, so that $d_1d_2 = d= \deg(\frakm)$.  Also, we let
$\K'$ be the subfield $\F[y_1]/\langle T_1(y_1) \rangle$ of $\K$, so
that $\K=\K'[y_2]/\langle T_2(\alpha_1,y_2)\rangle$; we have
$[\K':\F]=d_1$ and $[\K:\K']=d_2$.

\smallskip\noindent{{\em Step 1.}~~} By Proposition~\ref{prop:2}, we
can write $F = \sum_{0 \le i < d_2 \nu_t} F_i(x_1) x_2^i$, with all
$F_i$'s of degree at most $d_1 \mu_1$. Compute all
$F_i^*:=\pi_{T_1,\mu_1}(F_i) \in\K'[\xi_1]/\langle
{\xi_1}^{\mu_1}\rangle$, so as to obtain $G:=\sum_{0 \le i < d_2
  \nu_t} F_i^* x_2^i$. The cost of this step is $O(d_2 \nu_t \M(d_1
\mu_1)\log(\mu_1))$ operations in $\F$. Since $\nu_t \mu_1 \le
   {\mu}^2$ and $d_1 d_2 \mu =d \mu= n$, with $n = \deg(I)$, this is
   $O(\M(\mu n) \log(\mu))$.

\smallskip\noindent{{\em Step 2.}~~} Rewrite $G$ as $G=\sum_{i <
  \mu_1} G_i(x_2) \xi_1^i$, with all $G_i$'s in $\K'[x_2]$ of degree
at most $d_2 \nu_t$. Compute all $G_i^*:=\pi_{T_2,\nu_t}(G_i) \in
\K[\xi_2]/\langle {\xi_2}^{\nu_t} \rangle$.

To compute the $G_i^*$'s, we apply the univariate untangling algorithm
with coefficients in $\K'$ instead of $\F$.
The runtime of this second step is $O(\mu_1 \M(d_2 \nu_t)\log(\nu_t))$
operations $(+,\times)$ in $\K'$, which becomes $O(\mu_1 \M(d_1 d_2
\nu_t)\log(\nu_t))$ operations in $\F$, once we use Kronecker
substitution to do arithmetic in $\K'$. As for the first step, this is
$O(\M(\mu n) \log(\mu))$ operations in $\F$.

\smallskip\noindent{{\em Step 3.}~~} 
At this stage, we have $\sum_{i < d_2 \nu_t} G_i^* {\xi_1}^i \in
\K[\xi_2]/\langle {\xi_1}^{\mu_1}, {\xi_2}^{\nu_t} \rangle =
F(\xi_1+\alpha_1,\xi_2+\alpha_2) \bmod \langle {\xi_1}^{\mu_1},
{\xi_2}^{\nu_t} \rangle$. Discard all monomials lying in $J'$ and
return the result -- this involves no arithmetic operation.
On our example, the untangling
algorithm would pass from an ideal in $x_1,x_2$ (figure (a)  below) to
the monomial ideal $\langle \xi_1^2,\xi_2^2\rangle$ ({step~2},
figure (b)  below) then the monomial $\xi_1 \xi_2$ would be discarded to
get a result defined modulo $J'=\langle
\xi_1^2,\xi_1\xi_2,\xi_2^2\rangle$ ({step~3}, figure (c)  below).
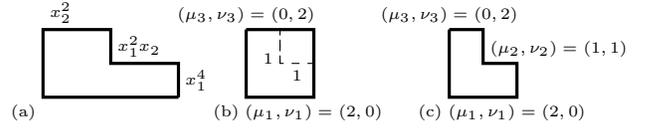
\vspace{1em}

\begin{tikzpicture}
  \fontfamily{times}{\fontsize{0.2cm}{0.2cm}\selectfont

      \coordinate [label={above right:$x_1^4$}] (A1) at (\xonestart, \xtwostart);
      \coordinate [] (B1) at (\xonestart, \xtwomiddle);
      \coordinate [label={above right:$x_1^2x_2$}] (C1) at (\xonemiddle, \xtwomiddle);
      \coordinate [] (D1) at (\xonemiddle, \xtwoend);
      \coordinate [label={above right:$x_2^2$}] (E1) at (\xoneend, \xtwoend);
      \coordinate [label={below left: (a)}] (O1) at (0, 0);
      \coordinate [] (F1) at (\xoneend, \xtwostart);
      \draw [very thick] (A1) -- (B1) -- (C1) -- (D1) -- (E1) -- (O1) -- (A1);
      
      \coordinate [label={above:$(\mu_3,\nu_3) = (0,2)$}] (A2) at (\muthree+\rOffset, \nuthree);
      \coordinate [] (B2) at (\mutwo+\rOffset, \nuthree);
      \coordinate [] (C2) at (\mutwo+\rOffset, \nutwo);
      \coordinate [] (D2) at (\muone+\rOffset, \nutwo);
      \coordinate [label={below:$(\mu_1,\nu_1) = (2,0)$}] (E2) at (\muone+\rOffset, \nuone);
      \coordinate [label={below left: (b)}] (O2) at (0+\rOffset, 0);
      \coordinate [] (F2) at (\muone+\rOffset, \nuthree);
      \draw [very thick] (A2) -- (B2) -- (F2) -- (D2) -- (E2) -- (O2) -- (A2);
      
      \draw [dashed] (C2) -- node [above left]  {$ $} (B2)
                         -- node [below left]  {$1$} (C2)
                         -- node [below ] {$1$} (D2)
                         -- node [above right] {$ $} (F2);

      \coordinate [label={above:$(\mu_3,\nu_3) = (0,2)$}] (A3) at (\muthree+2*\rOffset, \nuthree);
      \coordinate [] (B3) at (\mutwo+2*\rOffset, \nuthree);
      \coordinate [label={above right:$(\mu_2,\nu_2) = (1,1)$}] (C3) at (\mutwo+2*\rOffset, \nutwo);
      \coordinate [] (D3) at (\muone+2*\rOffset, \nutwo);
      \coordinate [label={below:$(\mu_1,\nu_1) = (2,0)$}] (E3) at (\muone+2*\rOffset, \nuone);
      \coordinate [label={below left: (c)}] (O3) at (0+2*\rOffset, 0);
    
      \draw [very thick] (A3) -- (B3) -- (C3) -- (D3) -- (E3) -- (O3) -- (A3);
  }
\end{tikzpicture}


\vspace{-1em} \subsection{Recursive tangling for monomial ideals}

The ideas used to perform univariate tangling, that is, to invert
$\pi_{T,\mu}$, carry over to bivariate situations. In this section, we
discuss the first of them, namely, a bivariate version of van der
Hoeven and Lecerf's recursive algorithm.  We still work under
assumption ${\bf H}_3$ that $J'$ is a monomial ideal.  As before,
$\mathcal{B}$ is the monomial basis of $\F[x_1,x_2]/I$ induced by the
Gr\"obner basis exhibited in Proposition~\ref{prop:2}.

\begin{proposition}\label{prop:piI}
  Under ${\bf H}_2$ and ${\bf H}_3$, given $G$ in
  $\K[\xi_1,\xi_2]/J'$ one can compute ${\pi_{\frakm,J'}}^{-1}(G)$
  using either $O(\M(d n) \log(n) + \M(n) \log(n)^2)$, or $O(\M(\mu
  n)\log(n)^2)$ operations in $\F$.  In particular, this can be done
  in $O(\M(n^{1.5}) \log(n)^2)$ operations.
\end{proposition}

As in~\cite{HoLe17}, our procedure is recursive; the recursion here is
based on the integer $\mu_1$. Given $G$ in $\K[\xi_1,\xi_2]/J'$, we
explain how to find $F$ in $\F[x_1,x_2]/I$ such that
$\pi_{\frakm,J'}(F)=G$, starting from the case $\mu_1=1$. 

\paragraph{}\label{par:leaf} If $\mu_1=1$, the ideal $J'$ is of the form
$\langle \xi_1,{\xi_2}^{\nu_2}\rangle$, and $\pi_{\frakm,J'}$ maps
$F(x_1,x_2)$ to $G:=F(\alpha_1,\xi_2+\alpha_2) \bmod {\xi_2}^{\nu_2}$.
In this case, note that the degree $n$ of $I$ is simply $d_1d_2\nu_2$.

\smallskip\noindent{{\em Step 1.}~~} Apply our univariate tangling
algorithm to $G$ in the variable $x_2$ to compute $F(\alpha_1,x_2):=
\pi_{T_2,\nu_2}^{-1} (G)\in \K'[x_2]/\langle T_2^{\mu_2}\rangle$,
working over the field $\K'=\F[y_1]/\langle T_1(y_1) \rangle$ instead of
$\F$.
This takes $O(\M(d_2 \nu_2) \log(\nu_2) + \M(d_2) \log(d_2))$ operations
$(+,\times)$ in $\K'$, together with $O(d_2)$ inversions in
$\K'$. Using Kronecker substitution for multiplications, this results
in a total of $O(\M(d_1d_2 \nu_2) \log(\nu_2) + \M(d_1d_2)
\log(d_1d_2))$ operations in $\F$.  We will use the simplified upper
bound $O(\M(d_1 d_2 \nu_2) \log(d_1 d_2\nu_2))=O(\M(n)\log(n))$.

\smallskip\noindent{{\em Step 2.}~~} The polynomial $F$ has degree
less than $d_1$ in $x_1$ and $d_2 \nu_2$ in $x_2$; for such $F$'s,
knowing $F(\alpha_1,x_2) \in \K'[x_2]/\langle T_2^{\mu_2}\rangle$ is
equivalent to knowing $F(x_1,x_2)$ in $\F[x_1,x_2]$. Thus, we are
done.

\paragraph{} Assume now that $\mu_1 > 1$, let $G$ be in $\K[\xi_1,\xi_2]/J'$ and 
let $\bar \mu:=\lceil \mu_1/2\rceil$. The following steps closely
mirror Algorithm~9 in~\cite{HoLe17}. For the cost analysis, we let
$S(\frakm, J')$ be the cost of applying
$\pi_{\frakm,J'}$ (see Proposition~\ref{prop:pi}) and 
$T(\frakm, J')$ be the cost of the recursive algorithm for
${\pi_{\frakm,J'}}^{-1}$.

\smallskip\noindent{{\em Step 1.}~~}
Let ${\bar G}:= G \bmod {\xi_1}^{{\bar\mu}}$, and compute
recursively ${\bar F}:={\pi_{\frakm, J'_0}}^{-1}(\bar G)$, with $J'_0:=J' +
\langle {\xi_1}^{\bar\mu}\rangle$. This costs $T(\frakm, J'_0)$.

\smallskip\noindent{{\em Step 2.}~~} Compute $H:=(G-\pi_{\frakm,
  J'}(\bar F)) {\rm ~div~} {\xi_1}^{\bar \mu}$, where the div operator
maps ${\xi_1}^i$ to $0$ for $i < \bar \mu$ and to
${\xi_1}^{i-\bar \mu}$ otherwise. This costs $S(\frakm, J')$.

\smallskip\noindent{{\em Step 3.}~~} 
Define $W:=\xi_1/\pi_{\frakm,J'}(T_1) \in \K[\xi_1,\xi_2]/\langle
{\xi_1}^{\mu_1},{\xi_2}^{\mu_2} \rangle$.  Because $T_1(\alpha_1)=0$
and $T'_1(\alpha_1)\ne 0$ (by our separability assumption), $W$ is
well-defined. This costs $S(\frakm, J')$ for $\pi_{\frakm,J'}(T_1)$
and $O(\M(d_1\mu_1))$ for inversion (since it involves $\xi_1$ only),
which is $O(\M(n))$.

\smallskip\noindent{{\em Step 4.}~~} Compute recursively ${\bar
  E}:={\pi_{\frakm, J'_1}}^{-1} (W^{\bar \mu} H \bmod J'_1)$, where
$J'_1$ is the colon ideal $J' : {\xi_1}^{\bar \mu}$. Since $W$ depends only on $\xi_1$, a
multiplication by $W$, or one of its powers, is done coefficient-wise
in $\xi_2$, for $O(\M(n))$ operations in $\F$. Thus, the cost to
compute $W^{\bar \mu} H \bmod J'_1$ is $O(\M(n) \log(n))$; to this, we
add $T(\frakm, J'_1)$.

\smallskip\noindent{{\em Step 5.}~~} Return $F:=\bar F + {T_1}^{\bar \mu} \bar E$.
The product $ {T_1}^{\bar \mu} \bar E$  requires no reduction,
since all its terms are in $\mathcal{B}$. Proceeding coefficient-wise
with respect to $x_2$, and using super-additivity, it
costs $O(\M(n))$.

On our example, we have $J'= \langle\xi_1^2,\xi_1\xi_2,\xi_2^2
\rangle$ (a), Step 1 uses $J_0' = \langle \xi_1,\xi_2^2
\rangle$ (b) and Steps 2-5 work on the colon ideal $J_1' = \langle
\xi_1,\xi_2\rangle$~(c). 
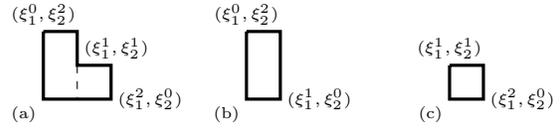
\vspace{1em}

\begin{tikzpicture}

  \fontfamily{times}{\fontsize{0.2cm}{0.2cm}\selectfont
      \coordinate [label={above:$(\xi_1^0,\xi_2^2)$}] (A1) at (\muthree, \nuthree);
     \coordinate [] (B1) at (\mutwo, \nuthree);
     \coordinate [label={above right:$(\xi_1^1,\xi_2^1)$}] (C1) at (\mutwo, \nutwo);
     \coordinate [] (D1) at (\muone, \nutwo);
     \coordinate [label={right:$(\xi_1^2,\xi_2^0)$}] (E1) at (\muone, \nuone);
     \coordinate [label={below left} : (a)] (O1) at (0, 0);
     \coordinate [] (M) at ( \mutwo,0);
      \draw [very thick] (A1) -- (B1) -- (C1) -- (D1) -- (E1) -- (O1) -- (A1);
      \draw [dashed] (C1) -- (M);
      
      \coordinate [label={above:$(\xi_1^0,\xi_2^2)$}] (A2) at (\muthree+\rOffset, \nuthree);
     \coordinate [] (B2) at (\mutwo+\rOffset, \nuthree);
     \coordinate [label={below left} : (b)] (O2) at (0+\rOffset, 0);
     \coordinate [label={right:$(\xi_1^1,\xi_2^0)$}] (M) at ( \mutwo+\rOffset,0);
      \draw [very thick] (A2) -- (B2) -- (M) -- (O2) -- (A2);
      
     \coordinate [label={above:$(\xi_1^1,\xi_2^1)$}] (C3) at (\mutwo+2*\rOffset-\mutwo, \nutwo);
     \coordinate [] (D3) at (\muone+2*\rOffset-\mutwo, \nutwo);
     \coordinate [label={right:$(\xi_1^2,\xi_2^0)$}] (E3) at (\muone+2*\rOffset-\mutwo, \nuone);
     \coordinate [label={below left:(c)}] (M3) at ( \mutwo+2*\rOffset-\mutwo,0);
      \draw [very thick] (C3) -- (D3) -- (E3) -- (M3) -- (C3);
  }
\end{tikzpicture}

\smallskip Let us justify that this algorithm is correct, by computing
$\pi_{\frakm,J'}(F)$, which is equal to $\pi_{\frakm, J'}(\bar F) +
\pi_{\frakm, J'}({T_1})^{\bar \mu}\pi_{\frakm, J'}(\bar E)\bmod
J'$. Note first that $\pi_{\frakm,J'}(\bar F) \bmod {\xi_1}^{\bar \mu}
= G\bmod {\xi_1}^{\bar \mu}$. Equivalently, $\pi_{\frakm,J'}(\bar F) =
G\bmod {\xi_1}^{\bar \mu} + {\xi_1}^{\bar \mu} (\pi_{\frakm,J'}(\bar
F) {\rm ~div~} {\xi_1}^{\bar \mu})$. Using the definition of $H$, this
is also $G\bmod {\xi_1}^{\bar \mu} + {\xi_1}^{\bar \mu}
(G {\rm ~div~} {\xi_1}^{\bar \mu} - H)$, that
is, $G -{\xi_1}^{\bar \mu} H$.
On the other hand, by definition of $\bar E$, we have
$$\pi_{\frakm,J'}(\bar E) = \pi_{\frakm, J'}({\pi_{\frakm,
    J'_1}}^{-1}(W^{\bar \mu} H \bmod J'_1)),$$ so that
$\pi_{\frakm,J'}(\bar E) \bmod J'_1 = W^{\bar \mu} H \bmod J'_1$. Now, $\pi_{\frakm, J'}(T_1)$ is a multiple of $\xi_1$, so
$\pi_{\frakm,J'}({T_1})^{\bar \mu}$ is a multiple of ${\xi_1}^{\bar
  \mu}$. Since ${\xi_1}^{\bar \mu} J'_1$ is in $J'$, we deduce
that $\pi_{\frakm, J'}({T_1})^{\bar \mu}\pi_{\frakm, J'}(\bar E)\bmod
J'$ is equal to $\pi_{\frakm, J'}({T_1})^{\bar \mu} W^{\bar \mu} H \bmod J'$,
and thus to ${\xi_1}^{\bar \mu} H$. Adding the two intermediate results so far,
we deduce that $\pi_{\frakm,J'}(F) = G$, as claimed.

Finally, we do the cost analysis. The runtime $\mathcal{T}(\frakm, J')$ satisfies the recurrence
relation
$$\mathcal{T}(\frakm, J') = \mathcal{T}(\frakm, J'_0) +
\mathcal{T}(\frakm, J'_1) + 
O(S(\frakm, J') + \M(n) \log(n)).$$
Using~\pref{par:leaf} and the super-linearity of $\M$, we 
see that the total cost at the leaves is $O(\M(n)\log(n))$.
Without loss of generality, we can assume that $S(\frakm, J')$ is 
super-linear, in the sense that 
$S(\frakm, J'_0)+ S(\frakm, J'_1) \le S(\frakm, J')$ holds at every level of the recursion.
Since the recursion has depth $O(\log(n))$, we get that
$\mathcal{T}(\frakm, J')$ is in $O(S(\frakm, J') \log(n)+ \M(n) \log(n)^2).$


\vspace{-1em}\subsection{Tangling for monomial ideals using duality}\label{ssec:dualTang2}

We finally present a bivariate analogue of the algorithm
introduced in Section~\ref{sec:univariate}. Since the runtimes obtained
are in general worse than those in the previous subsection, we
only sketch the construction.

All notation being as before, let $G$ be in $\K[\xi_1,\xi_2]/J'$, and
let $F \in \F[x_1,x_2]/I$ be such that $\pi_{\frakm,J'}(F) = G$.  Following ideas
from~\cite{NeRaSc17}, we now use several linear forms.  Thus, let
$\ell_1,\dots,\ell_\gamma$ be module generators of
$(\K[\xi_1,\xi_2]/J')^*$, where the ${}^*$ means that we look at
the dual of $\K[\xi_1,\xi_2]/J'$ as an $\F$-vector space. Define $\ell'_1:=
\mbox{$G\cdot \ell_1$},\dots,\ell'_\gamma:=G \cdot \ell_\gamma$, as well as
\begin{align*}
L_1:=\pi_{\frakm,J'}^\perp(\ell_1),\dots,L_\gamma:=\pi_{\frakm,J'}^\perp(\ell_\gamma)\\
L'_1:=\pi_{\frakm,J'}^\perp(\ell'_1),\dots,L'_\gamma:=\pi_{\frakm,J'}^\perp(\ell'_\gamma)
\end{align*}
in $(\F[x_1,x_2]/I)^*$. As in the one variable case, for $i=1,\dots,\gamma$
the relation $\pi_{\frakm,J'}(F)\cdot \ell_i = \ell'_i$ implies that $F \cdot L_i
= L'_i$.

The first question is to determine suitable
$\ell_1,\dots,\ell_\gamma$. Consider generators
$\xi_1^{\mu_1}\xi_2^{\nu_1},\dots,\xi_1^{\mu_t}\xi_2^{\nu_t}$ of $J'$,
with the $\mu_i$'s decreasing and $\nu_i$'s increasing as before. For
$i=1,\dots,t-1$, define $\ell_i$ by $\ell_i(\alpha_1^{d_1-1}
\alpha_2^{d_2-1} {\xi_1}^{\mu_i-1} {\xi_2}^{\nu_{i+1}-1}) =1$, all
other $\ell_i(\alpha_1^{e_1} \alpha_2^{e_2} {\xi_1}^{r_1}
{\xi_i}^{r_2})$ being set to zero. Then, following
e.g.~\cite[Section~21.1]{Eisenbud96}, one verifies that these
linear forms are module generators of $(\K[\xi_1,\xi_2]/J')^*$.

As in the univariate case, we can compute all $L_i$ and $L'_i$ by
transposing the untangling algorithm, incurring $O(t)$ times the cost
reported in Proposition~\ref{prop:piI}. Then, it remains to solve all
equations $F \cdot L_i = L'_i$, $i=1,\dots,t-1$ (this 
system is not square, unless $t=2$). 
We are not aware of a quasi-linear time algorithm to solve such
systems.  The matrix of an equation such as $F \cdot L_i = L'_i$ is
sometimes called {\em multi-Hankel}~\cite{BeBoFa15}. It can be solved using structured linear algebra
techniques~\cite{BeBoFa15} (Here, we have several such systems to solve at once;
this can be dealt with as in~\cite{ChJeNeScVi15}). As
in~\cite{BeBoFa15}, using the results from~\cite{BoJeMoSc17} on
structured linear system solving, we can find $F$ in Monte
Carlo time $O((st)^{\omega-1}\M(t n) \log(t n))$, with $s:=\min(\mu_1,
\nu_t)$, where $\omega$ is the exponent of linear algebra (the best
value to date is $\omega \le 2.38$~\cite{CoWi90,LeGall14}).
Thus, unless both $s$ and $t$ are small, the overhead induced by the
linear algebra phase may make this solution inferior to the one in
the previous subsection.

\vspace{-1em}
\subsection{An Application}

To conclude, we describe a direct application of our results to the
complexity of multiplication and inverse in $\A:=\F[x_1,x_2]/I$: under
assumptions ${\bf H}_2$ and ${\bf H}_3$, both can be done in the time
reported in Proposition~\ref{prop:piI}, to which we add $O(\M(n)
\log(n)^3)$ in the case of inversion.  Even though the algorithms are
not quasi-linear time in the worst case, to our knowledge no previous
non-trivial algorithm was known for such operations.

The algorithms are simple: untangle the input, do the multiplication,
resp.\ inversion, in $\A':=\K[\xi_1,\xi_2]/J'$, and tangle the result.
The cost of tangling dominates that of untangling. The appendix below
discusses the cost of arithmetic in $\A'$: multiplication and inverse
take respectively $O(\M(\mu)\log(\mu))$ and $O(\M(\mu)\log(\mu)^2)$
operations $(+,-,\times)$ in $\K$, plus one inverse in $\K$ for the
latter. Using Kronecker substitution, the runtimes become
$O(\M(n)\log(n))$ and $O(\M(n)\log(n)^2)$ operations in $\K$, with $n
= \deg(I)$; this is thus negligible in front of the cost for tangling.


\vspace{-1em}
\section*{Appendix: Bivariate power series arithmetic}

We prove that for a field $\F$ and zero-dimensional monomial ideal $I
\subset \F[x_1,x_2]$, multiplication and inversion in $\F[x_1,x_2]/I$
can be done in softly linear time in $\delta:=\deg(I)$,
starting with multiplication.  

For an ideal such as $I = \langle x_1^{\mu}, x_2^\nu \rangle$, the
claim is clear. Indeed, to multiply elements $F$ and $G$ of
$\F[x_1,x_2]/I$ we multiply them as bivariate polynomials and discard
unwanted terms.  Bivariate multiplication in partial degrees less than
$\mu$, resp.\, $\nu$, can be done by Kronecker substitution in time
$O(\M(\mu \nu))=O(\M(\delta))$, which is softly linear in $\delta$, as
claimed. However, this direct approach does not perform well for cases
such as $I = \langle x_1^\mu, x_1 x_2, x_2^\nu\rangle$: in this case,
for $F$ and $G$ reduced modulo $I$, the product $FG$ as polynomials
has $\mu \nu$ terms, but $\delta = \mu + \nu-1$. The following result
shows that, in general, we can obtain a cost almost as good as in the
first case, up to a logarithmic factor. Whether this extra factor can
be removed is unclear to us. In the rest of this appendix, we write $I =
\langle x_1^{\mu_1}x_2^{\nu_1}, x_1^{\mu_2} x_2^{\nu_2}, \dots,
x_1^{\mu_t}x_2^{\nu_t} \rangle$, with $\mu_i$'s decreasing, $\nu_i$'s
increasing and $\nu_1=\mu_t=0$.

\begin{proposition}
  Let $I$ be a zero-dimensional monomial ideal in $\F[x_1,x_2]$
  of degree $\delta$.  Given $F,G$ reduced modulo $I$, one can compute
  $FG\bmod I$ in $O(\M(\delta)\log(\delta))$ operations $(+,-,\times)$ in $\F$.
\end{proposition}

\vspace{-1mm}

\noindent{\bf A.1.}  We start by giving an algorithm
of complexity $O(t \M(\delta))$ for multiplication modulo $I$.
Let $F$ and $G$ be two polynomials reduced modulo $I$. To compute $H:=FG
\bmod I$ it suffices to compute $H_i:=FG \bmod \langle x_1^{\mu_i},
x_2^{\nu_{i+1}}\rangle$ for $i = 1,\dots,t-1$; all monomials in $H$
appear in one of the $H_i$'s (some of them in several $H_i$'s).  We
saw that multiplication modulo $ \langle x_1^{\mu_i},
x_2^{\nu_{i+1}}\rangle$ takes $O(\M(\mu_i \nu_{i+1}))$ operations in
$\F$, which is $O(\M(\delta))$, so the total cost is $O(t
\M(\delta))$.

\smallskip\noindent{\bf A.2.}
In the general case,  define
$i_1:=1$. We let $i_2 \le t$ be the smallest index greater than $i_1$
and such that $\mu_{i_2} < \mu_{i_1}/2$, and iterate the process to
define a sequence $i_1 = 1 < i_2 < \cdots < i_s=t$. The ideal $I'$ is
then defined by the monomials
$x_1^{\mu_{i_1}}x_2^{\nu_{i_1}},\dots,x_1^{\mu_{i_s}}x_2^{\nu_{i_s}}$.
By construction, $I$ contains $I'$; hence, to compute a product 
modulo $I$, we may compute it modulo $I'$ and discard unwanted terms.

Multiplication modulo $I'$ is done using the algorithm of {\bf A.1},
in time $O(s \M(\delta'))$, with $\delta':=\deg(I')$.  Hence, we need
to estimate the degree $\delta'$ of $I'$, as well as its number of
generators $s$.

The degree $\delta$ of $I$ can be written as $\sum_{r=1}^{s-1}
\sum_{i=i_r}^{i_{r+1}-1} \mu_i (\nu_{i+1}-\nu_i)$; this is simply
counting the number of standard monomials along the rows. For a given
$r$, all indices $i$ in the inner sum are such that $\mu_i \ge
\mu_{i_r}/2$, so the sum is at least $1/2 \sum_{r=1}^{s-1} \mu_{i_r}
(\nu_{i_{r+1}}-\nu_{i_r})$, which is  the degree of $I'$.
Hence, $\delta \ge 1/2 \delta'$, that is, $\delta'\le 2 \delta$.  To
estimate the number $s$,  the inequalities $\mu_{i_{r+1}} <
\mu_{i_r}/2$ for all $r \le s$ imply that $\mu_{i_{s-1}} <
\mu_1/2^s$. We deduce that $2^s \le \mu_1/\mu_{i_{s-1}} \le \mu_1$
(since $\mu_{i_{s-1}} \ge 1$), which itself is at most $\delta$.
Thus, $s \in O(\log(\delta))$. Overall, the cost of multiplication modulo
$I'$, and thus modulo $I$, is $O(\M(\delta) \log(\delta))$.

\begin{corollary}
  For $I$ as in the previous proposition and $F$ reduced modulo $I$,
  with $F(0,0) \ne 0$, $1/F \bmod I$ can be computed in
  $O(\M(\delta)\log(\delta)^2)$ operations $(+,-,\times)$ in $\F$, 
  and one inverse.
\end{corollary}

\vspace{-1mm}

\noindent{\bf A.3.} We proceed by
induction using Newton iteration.  If $\mu_1 =1$ then $I
= \langle x_1,x_2^{\nu_2}\rangle$, so inversion modulo $I$ is
inversion in $\F[x_2]/\langle x_2^{\nu_2}\rangle$. It can be done in
time $O(\M(\delta))$ using univariate Newton iteration, involving only
the inversion of the constant term of the input.

Otherwise, define $\bar\mu:= \lceil \mu_1/2\rceil$, and let $\bar I$ be the ideal
with generators  $ x_1^{\bar\mu}, x_1^{\mu_2} x_2^{\nu_2}, \dots,
x_2^{\nu_t}$ (all monomials in this list with $\mu_i \ge \bar\mu$ may be
discarded). Given $F$ in $\F[x_1,x_2]/I$, we start by computing the
inverse of $\bar G$ of $\bar F:=F \bmod \bar I$ in $\F[x_1,x_2]/\bar I$. Since ${\bar I}^2$
is contained in $I$, knowing $\bar G$, one step of Newton iteration allows
us to compute $G:=1/F \bmod I$ as $G = 2\bar G -\bar G^2 F \bmod I$.
Using the previous proposition, we deduce $G$ from $\bar G$ in
$O(\M(\delta)\log(\delta))$ operations. We repeat the 
recursion for $O(\log(\delta))$ steps, and the degrees of the ideals 
we consider  decrease, so the overall runtime is 
$O(\M(\delta)\log(\delta)^2)$.
 
\bibliographystyle{plain}   
\bibliography{bib}
  
\end{document}